\documentclass[final]{siamltex}

\usepackage{amsmath,amssymb,graphicx,color}

\newcommand{\ds}{\displaystyle}

\def\R{{\Bbb R}}
\def\Om{\Omega}
\def\om{\omega}

\def\p{\partial}
\def\q{\quad}
\def\na{\nabla}

\def\k{{\mathbf k}}
\def\h{{\mathbf h}}
\def\n{\mathbf{n}}
\def\u{\mathbf{u}}
\def\v{\mathbf{v}}

\def\x{\mathbf{x}}
\def\y{\mathbf{y}}
\def\z{\mathbf{z}}
\def\F{{\mathbf F}}
\def\mE{{\mathcal E}}
\def\eref#1{(\ref{#1})}
\def\g{{\mathbf g}}

\def\bq{{\mathbf q}}
\def\btau{\boldsymbol{\tau}}
\def\bGamma{\boldsymbol{\Gamma}}
\def\zt{\z^*}

\title{Mathematical modeling of mechanical vibration assisted
conductivity imaging\thanks{\footnotesize Ammari was supported by
the ERC Advanced Grant Project MULTIMOD--267184. Eunjung Lee was
supported by Basic Science Research Program through the NSF of Korea
2013R1A1A1004836. Kwon and Seo were supported by the National
Research Foundation of Korea (NRF) grant funded by the Korean
government (MEST) (No. 2011-0028868, 2012R1A2A1A03670512). }}

\author{Habib Ammari\thanks{Department of Mathematics and Applications,
Ecole Normale Sup\'erieure, 45 Rue d'Ulm, 75005 Paris, France ({\tt
habib.ammari@ens.fr}).}\and  Eunjung Lee\thanks{Department of
Computational Science and Engineering, Yonsei University, 50
Yonsei-Ro, Seodaemun-Gu, Seoul 120-749, Republic of Korea ({\tt
eunjunglee@yonsei.ac.kr}, {\tt 3c273-85@hanmail.net}, {\tt
seoj@yonsei.ac.kr}). Eunjung Lee is the corresponding author.} \and
Hyeuknam Kwon\footnotemark[3] \and Jin~Keun~Seo\footnotemark[3] \and
Eung~Je~Woo\thanks{\footnotesize Department of Biomedical
Engineering, College of Electronics and Information, Kyung Hee
University,  1732 Deogyeong-daero, Giheung-gu Yongin-si, Gyeonggi-do
446-701, Republic of Korea ({\tt ejwoo@khu.ac.kr}).}}

\begin{document}

\maketitle

\begin{abstract}
This paper aims at mathematically modeling a new multi-physics conductivity
imaging system incorporating mechanical vibrations simultaneously applied to an imaging object together with current injections. We perturb the internal conductivity distribution by applying time-harmonic mechanical vibrations on the boundary. This enhances the effects of any conductivity discontinuity on the induced internal current density distribution. Unlike other conductivity contrast enhancing frameworks, it does not require a prior knowledge of a reference data. In this paper, we provide a mathematical framework for this novel imaging modality. As an application of the vibration-assisted impedance imaging framework, we propose a new breast image reconstruction method in electrical impedance tomography (EIT). As its another application, we investigate a conductivity anomaly detection problem and provide an
efficient location search algorithm. We show both analytically and numerically that the applied mechanical vibration increases the data sensitivity to the conductivity contrast and enhances the quality of reconstructed images and anomaly detection results. For numerous applications in impedance imaging, the proposed multi-physics method opens a new difference imaging area called the vibration-difference imaging, which can augment the time-difference and also frequency-difference imaging methods for sensitivity improvements.
\end{abstract}



\pagestyle{myheadings} \thispagestyle{plain}

\section{Introduction}\label{sec:intro}

Electrical impedance imaging methods typically use a number of
surface electrodes to inject currents and measure induced voltages. There exists numerous electrode configurations in electrical impedance tomography (EIT) for applications to the head, chest, breast, abdomen, and limbs. For the breast imaging, it is often desirable to adopt a planner probe or a plate with an array of electrodes placed over the breast. For this kind of electrode configuration, we may also investigate an electrical admittance imaging method, where exit currents subject to an applied voltage are measured.

Though the time-difference EIT has been finding clinical applications especially for the lungs, there still exist technical difficulties in producing reliable images of tumor and stroke, where the time-difference approach is not feasible. Noting that a reference data obtained beforehand at a certain time in the absence of an anomaly is not available, we need to develop a novel method to image the anomaly if it exists without using a time-difference approach.

Frequency-difference methods have been, therefore, suggested to utilize different frequency-dependent electrical properties of anomalies such as tumor and stroke compared with the background region. Although the frequency-difference methods have their own
advantages, reconstructions of reliable frequency-difference images are hindered from several technical difficulties: the low sensitivity of measured voltage data on a local change of frequency-dependent admittivity, the inhomogeneous frequency-dependent admittivity distribution of the background region, measurement errors especially at high frequencies due to stray capacitances, and so on.

To improve the sensitivity of the frequency-difference EIT for tumor and stroke imaging applications, we propose redesigning the measurement method to include a process of mechanical vibration. By augmenting the primary physical process of electrical conduction, it allows us to perform a vibration-difference imaging in addition to the frequency-difference imaging.

In this paper, we will first describe the mathematical framework of the vibration-difference method in EIT. Emphasizing the sensitivity improvement by the conductivity modulation through a mechanical vibration, we will suggest two measurement settings where this new approach can be adopted. We will carry out rigorous analyses of these new vibration-difference methods and show their performances and feasibilities through numerical experiments.

\section{Mathematical modeling }\label{sec:model}

This section provides  a new mathematical model for a mechanical
vibration assisted conductivity imaging and its theoretical ground. To propose the model, we set $\Omega$ to be a bounded domain with a boundary $\partial\Omega$ of class $\mathcal{C}^2$ in $\mathbb{R}^3$. We assume that the electrical conductivity $\sigma$
of $\Omega$ is of class $\mathcal{C}^2(\overline{\Omega})$.
Moreover, there exist $\underline{\sigma}$ and $\overline{\sigma}$
such that $ 0< \underline{\sigma}<\sigma<\overline{\sigma} <\infty.$
Furthermore, we suppose that $\sigma$ is constant on a neighborhood
of the boundary $\partial \Omega$.

When we inject a current $q\in \mathcal{C}^{1,\alpha}(\partial
\Omega)$ for some $0<\alpha<1$ with its mean-value of zero, $\int_{\partial
\Omega} q = 0$,   the resulting electrical potential $v$ is governed
by the following conductivity equation:
\begin{equation}\label{eq:potential_homo}
\left\{\begin{array}{rcll}
\nabla\cdot(\sigma\nabla v)&=&0&\mbox{in}~\Omega,\\
\sigma\frac{\partial v}{\partial \n}&=&q&\mbox{on}~\partial\Omega,\\
\ds \int_{\partial \Omega} v = 0,
\end{array}\right.
\end{equation}
where $\sigma \partial v/ \partial \n = \n \cdot (\sigma \nabla)$
with $\n$ being the outward unit normal vector at $\partial \Omega$.

To perturb the conductivity distribution $\sigma$, we
attach a mechanical vibrator on the boundary $\p\Om$ and apply a
time-harmonic vibration. We assume that $\Omega$ is composed of a
linearly elastic, isotropic, and incompressible material of density
equal to $1$. We let $\mu$ be the shear modulus of $\Omega$. We
assume that $\mu$ belongs to
$\mathcal{C}^{0,\alpha}(\overline{\Omega})$ and there exist
$\underline{\mu}$ and $\overline{\mu}$ such that  $ 0<
\underline{\mu}<\mu<\overline{\mu} <\infty.$ If $\om$ is the
operating angular frequency, the resulting time-harmonic
elastic displacement is denoted as $\underline{\u}(\x,t)= \rm{Re} \{ e^{i
\omega t} \u(\x) \}$ for $\x \in \Omega$ and $t \in \R^+$, where
$\u$ satisfies the Stokes system \cite{Habib:2011}:
\begin{equation}\label{eq:ham-elasticity}
\left\{\begin{array}{rcll}
\om^2\u+\na\cdot (\mu (\na \u +\na {\u}^T)) -\nabla p &=&0 &\mbox{in}~\Omega, \\
\nabla \cdot \u &=&0 &\mbox{in}~\Omega,\\
\u &=& {\bf g}  & \mbox{on }
\partial\Omega.
\end{array}\right.
\end{equation}
 Here,
$\na {\u}^T$ denotes the transpose of the matrix $\na
{\u}$ and ${\bf g} \in \mathcal{C}^{1,\alpha}(\partial \Omega)$ is such that the compatibility condition
$\int_{\partial \Omega} {\bf g} \cdot \n = 0$ holds.

In the sequel, we assume that $-\omega^2$ is not a Dirichlet eigenvalue of the Stokes system on $\Omega$; see \cite{zuazua}. We also recall that the analytical continuation principle holds true for the Stokes system. In fact, it can be proved using the results in \cite{fabre1, fabre2, nakamura} that if $\u$ is zero in a ball inside $\Omega$, then $\u$ is identically zero everywhere in $\Omega$ provided that $\mu \in \mathcal{C}^{0,1}(\overline \Omega)$.   Moreover, from \cite{chen, giaquinta, yanyan},  $\u \in \mathcal{C}^{1,\alpha}(\overline{\Omega})$ and there exists a positive constant $C$ depending only on $\mu, \omega,$ and $\Omega$ such that
$$
|| \u ||_{\mathcal{C}^{1,\alpha}(\overline{\Omega})} \leq C || {\bf
g}||_{\mathcal{C}^{1,\alpha}({\partial \Omega})}.$$
The displacement $\underline{\u}$ causes the perturbation of the
conductivity distribution, $\sigma^\diamond$, which can be described
as, for a time $t\in\mathbb{R}^+$,
\begin{equation}\label{eq:sigma_ori}
\sigma^\diamond(\x+\underline{\u}(\x,t),t)=\sigma(\x),\quad \forall
\x\in \Omega.
\end{equation}
It induces $\u\cdot\nabla \sigma$, which can be captured by various
electrical impedance imaging techniques. To show this, we let
$\Omega^\diamond=\{\x+ \underline{\u}(\x,t)\,|\, \ \x\in
\Omega,\,\mbox{for a time}~ t \in \R^+ \}$. We can rewrite the relation
(\ref{eq:sigma_ori}) as
\begin{equation} \label{link}
\begin{aligned} \sigma^\diamond =\sigma \circ \left( \mathbb{I} + \underline{\u} \right)^{-1},
\quad  \sigma= \sigma^\diamond  \circ \left( \mathbb{I} +
\underline{\u} \right), \quad \x \in \Omega^\prime, t \in \R^+,
\end{aligned}
\end{equation}
where $\mathbb{I} + \underline{\u}$ is a map such that, for a time
$t\in\mathbb{R}^+$,
$$
\mathbb{I} + \underline{\u}(\cdot,t):\x\mapsto (\x+\u(\x,t),t)
$$
and $\Omega^\prime\subset \subset \left(\Omega\cap
\Omega^\diamond\right)$ is any smooth simply connected domain.

Assuming that $\mu \in \mathcal{C}^{0,\alpha}(\overline \Omega)$,
$\|{\u}\|_{L^\infty(\Omega)}\ll 1$ and $\sigma \in
\mathcal{C}^2(\overline \Omega)$, the perturbed conductivity can be
approximated as
\begin{equation}\label{eq:sigma}
\sigma^\diamond(\x,t)\approx \sigma(\x)- \underline{\u}(\x,t)
\cdot\nabla\sigma(\x), \quad \x \in \Omega',~ t \in \R^+,
\end{equation}
since $\sigma$ is assumed to be constant on a neighborhood of the
boundary $\partial \Omega$.
Let $\underline v^\diamond$ denote the electrical potential of
\eref{eq:potential_homo} with the conductivity distribution $\sigma^\diamond$ in place of $\sigma$. The potential $\underline v^\diamond$
varies with the time-change of $\sigma^\diamond$:
\begin{equation}\label{eq:pert_u}
\nabla\cdot\left((\sigma(\x)-\underline{\u}(\x,t)\cdot\nabla\sigma(\x))\nabla
\underline v^\diamond(\x,t)\right)~\approx~0\q\mbox{for}~\x\in\Om, t
\in \R^+.
\end{equation}
Denoting $\underline {v_1}(\x,t):=\underline v^\diamond(\x,t)-v(\x)$, we have
\begin{eqnarray}\label{approx-v1}
\nabla\cdot (\sigma(\x)\na \underline {v_1}(\x,t))
&=&\nabla\cdot\left( \underline{\u}(\x,t)\cdot\nabla\sigma(\x) ~\na
v(\x) \right) + \nabla\cdot \left(
\underline{\u}(\x,t)\cdot\nabla\sigma(\x) ~\na \underline{v_1}(\x,t)\right)\nonumber  \\
&\approx &\nabla\cdot \left(
\underline{\u}(\x,t)\cdot\nabla\sigma(\x) ~\na v(\x) \right).
\end{eqnarray}
In the last approximation, we dropped $\nabla
\hspace{-0.1cm}\cdot\hspace{-0.1cm}
\left(\underline{\u}(\x,t)\cdot\nabla\sigma(\x) \na
\underline{v_1}(\x,t) \right)$ since both $\underline{\u}$ and
$\underline{v_1}$ are small.

From \eref{eq:potential_homo}, by virtue of the approximation
 \eref{approx-v1}, it follows that $\underline{v_1}(\x,t)$ satisfies
\begin{equation}\label{eqn-for-v1-bar}
\left\{\begin{array}{rcll}
\nabla\cdot(\sigma(\x)\nabla \underline{v_1}(\x,t))&=&
{\rm Re} \{  e^{i\om t}\nabla\cdot(\u(\x)\cdot\nabla \sigma(\x) ~\nabla v(\x)) \} &
\mbox{for}~(\x,t)\in \Omega\times \R^+,\\
\sigma\frac{\partial
\underline{v_1}}{\partial\n}&=&0&\mbox{on}~\partial\Omega \times
\R^+.
\end{array}\right.
\end{equation}
Therefore, we can express $\underline{v_1}$ as
$$
\underline{v_1}(\x,t)~=~{\rm Re} \{ e^{i\omega t} v_1(\x)  \},
$$ where $v_1$ is the solution to the following conductivity equation:
\begin{equation}\label{eqn-for-v1}
\left\{\begin{array}{rcll}\nabla\cdot(\sigma\nabla v_1)&=&
\nabla\cdot((\u\cdot\nabla \sigma)\nabla v)&
\mbox{in}~\Omega,\\
\sigma\frac{\partial v_1}{\partial\n}&=&0&\mbox{on}~\partial\Omega.
\end{array}\right.
\end{equation}

Finally, we arrive at
\begin{equation} \label{approx}
\underline v^\diamond(\x,t) \approx v(\x) + {\rm Re} \{ e^{i \omega
t}{v_1}(\x) \}, \quad \x \in \Omega, t \in \mathbb{R}^+,
\end{equation}
where ${v_1}$ is the solution to (\ref{eqn-for-v1}). Note that the
measured data over time yields the knowledge of $v_1$, which is
(approximately) the difference between $v$ and $\underline v^\diamond$ measured without and with the mechanical vibration, respectively. The equation (\ref{eqn-for-v1}) clearly shows
that $v_1$ carries information of $\u\cdot\nabla\sigma$. The major advantage of the proposed method is then to extract the additional information of $\u\cdot\nabla\sigma$ from the boundary current-voltage relation.

In the following sections, we will deal with two imaging problems. The first one is to visualize the projected image of $\u\cdot\nabla\sigma$ in a breast imaging setting using a planar array of electrodes. The second is for the anomaly identification in a more general EIT system configuration. Based on the approximation (\ref{approx}), we will provide novel reconstruction
methods with rigorous analyses. We will extend the approximation (\ref{approx}) to piecewise constant conductivity distributions.

\section{Vibration-assisted electrical impedance imaging}\label{sec:PIR}

Quantitative measurements and imaging of the electrical and
mechanical tissue properties have been studied to improve the
sensitivity and specificity of the X-ray mammography for early
diagnosis of breast cancer. These methods are motivated from the
experimental findings that electrical and mechanical properties of
the breast tissue change with its pathological state. Malignant
breast tissues have higher electrical conductivity values and are
significantly stiffer than the surrounding normal tissues
\cite{hartov, jossinet, silva}.

For example, for adjunct uses with the X-ray mammography, T-scan or trans-admittance scanner
(TAS) uses a probe with an array of current-sensing electrodes to measure exit currents
induced by an applied voltage \cite{trans, assenheimer, kim, seo}. In trans-admittance mammography
(TAM), the breast is compressed by a pair of parallel plates just
like the X-ray mammography configuration \cite{TAM:2012}. One plate is a solid
conductor connected to a constant voltage source and the other is
equipped with an array of current-sensing electrodes to measure the
exit currents induced by the applied voltage.

To apply the idea of the vibration-assisted impedance imaging technique, we consider the breasting imaging setup in Figure \ref{fig:new_cont}. The imaging domain $\Omega$ is a hexahedral container, and two pairs of driving electrodes ($\mE_{1\pm}$ and
$\mE_{2\pm}$) are attached  on the side of the container  to inject
two linearly independent currents. On the bottom plate denoted by
$\Upsilon$, there are many small voltage-sensing electrodes to measure the distribution of the induced voltages. One may increase the number of measurements by increasing the number of the voltage-sensing electrodes. Though it is an impedance imaging configuration, we may apply this vibration-assisted method to TAS and TAM, which are based on the admittance imaging configuration.

\begin{figure}[ht!]
\begin{center}
\includegraphics[scale=0.4]{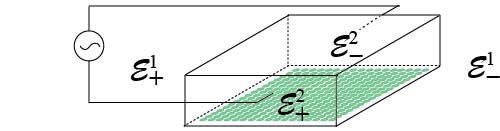}
\caption{Basic setup of a vibration-assisted electrical impedance
imaging system}\label{fig:new_cont}
\end{center}
\end{figure}

When we inject a current with its amplitude $I$ through each pair of the driving electrodes, the
resulting potential $v_{j,0}~(j=1,2)$ satisfies the following mixed boundary value problem:
\begin{equation}\label{mainEq}
\left\{
\begin{array}{ll}
\nabla \cdot(\sigma\nabla v_{j,0})=0 &\mbox{in} \q\Omega,\\
\int_{\mE_{j\pm}} \sigma \frac{\partial v_{j,0}}{\partial \n}=\pm\emph{ I},
~~\int_{\mE_{k\pm}} \sigma \frac{\partial v_{j,0}}{\partial \n}=0&\mbox{for}
\q k\neq j\q\q(k=1,2),\\
\n\times \nabla v_{j,0}=0&\mbox{on}\q \mE_{1\pm}\bigcup\mE_{2\pm},\\
\sigma \frac{\partial v_{j,0}}{\partial \n}=0&\mbox{on}\q
\partial\Omega \backslash(\mE_{1\pm}\bigcup\mE_{2\pm}).
\end{array}
\right.
\end{equation}
In this model, effects of the contact impedance on the measured voltage
values are neglected since we measure the voltages only on the voltage-sensing electrodes.

Attaching a mechanical vibrator on a part of $\p\Om$, we follow the
approach described in the previous section. The time-harmonic displacement
$\u$  inside $\Om$ generates  the conductivity perturbation
$\u\cdot\nabla \sigma$. Following the derivation to obtain $v_1$ in
\eref{eqn-for-v1}, the difference in the potentials between two
cases of without and with the applied mechanical vibration can be described by
$v_{j,1}$ for each $j=1,2$:
\begin{equation}\label{mainEq2}
\left\{
\begin{array}{ll}
\nabla\cdot(\sigma\nabla v_{j,1})=
\nabla\cdot((\u\cdot\nabla \sigma)\nabla v_{j,0})& \mbox{in} \q\Omega,\\
\int_{\mE_{k\pm}} \sigma \frac{\partial v_{j,1}}{\partial \n}=0&\mbox{for}
\q k=1,2,\\
\n\times \nabla v_{j,1}=0&\mbox{on}\q \mE_{1\pm}\bigcup\mE_{2\pm},\\
\sigma \frac{\partial v_{j,1}}{\partial \n}=0&\mbox{on}\q
\partial\Omega \backslash(\mE_{1\pm}\bigcup\mE_{2\pm}).
\end{array}
\right.
\end{equation}

Now, we are ready to provide a novel projective image reconstruction
method for imaging the conductivity perturbation $\u\cdot\nabla
\sigma$. From \eref{mainEq2}, we have
$$
\nabla\cdot(\sigma\nabla v_{j,1})= \nabla\cdot\left(
(\u\cdot\nabla\ln\sigma) (\sigma\nabla v_{j,0})\right)
=\left(\nabla (\u\cdot\nabla\ln\sigma)\right)\cdot (\sigma\nabla
v_{j,0}),~~j=1,2
$$
and, therefore,
\begin{equation}\label{formula}
\left[
  \begin{array}{c}
    \sigma\nabla v_{1,0}^T \\
   \sigma\nabla v_{2,0}^T\\
  \end{array}
\right] [\nabla\left(\u\cdot\nabla\ln \sigma\right)]^T= \left[
  \begin{array}{c}
   \nabla\cdot(\sigma\nabla v_{1,1}) \\
   \nabla\cdot(\sigma\nabla v_{2,1})\\
  \end{array}
\right] \q\q\mbox{in}~\Omega.
\end{equation}
Noting that $\Upsilon$, the bottom plate, is the measuring surface,
the measured quantities are
$$
V_{j,0}=v_{j,0}|_{\Upsilon}\q\mbox{and}\q
V_{j,1}:=v_{j,1}|_{\Upsilon}\q (j=1,2).
$$
Then, the projective image of $\left[\u\cdot\nabla\ln
\sigma\right]_{\mbox{\tiny proj}}$ is obtained by solving the
following two-dimensional Poisson equation
\begin{equation}\label{project-formula}
\nabla_{\mbox{\tiny xy}}^2\left(\left[\u\cdot\nabla\ln
\sigma\right]_{\mbox{\tiny proj}}^T\right)=\nabla_{\mbox{\tiny
xy}}\cdot\left(\left[
  \begin{array}{c}
    \sigma\nabla_{\mbox{\tiny xy}} V_{1,0}^T \\
   \sigma\nabla_{\mbox{\tiny xy}} V_{2,0}^T\\
  \end{array}
\right]^{-1}
\left[
  \begin{array}{c}
   \nabla_{\mbox{\tiny xy}}\cdot(\sigma\nabla_{\mbox{\tiny xy}} V_{1,1}) \\
   \nabla_{\mbox{\tiny xy}}\cdot(\sigma\nabla_{\mbox{\tiny xy}} V_{2,1})\\
  \end{array}
\right]\right) ~~\mbox{on}~\Upsilon
\end{equation}
with the homogeneous Neumann boundary condition on $\p\Upsilon$. We will show from Numerical simulations that the formula \eref{project-formula}
provides a diffused image of the projected distribution of
$\u\cdot\nabla\ln \sigma$ onto the plate $\Upsilon$.

For the breast imaging method to be compatible with the X-ray mammography, we may apply the vibration to the plate on the top to compress the breast. The mechanical vibrator connected to the top plate will create propagating vibration
waves inside the breast. The time-harmonic mechanical waves induce
time-harmonic displacements of the breast tissues to result in
time-harmonic conductivity vibrations. It is crucial to observe that
these time-harmonic conductivity variations perturb the internal
conductivity distribution and enhance the effects of the
conductivity contrast on the internal current density distribution.
Vibrating the breast during the data acquisition, therefore, we can further highlight the conductivity contrast between normal and cancerous tissues.

\section{Vibration-assisted anomaly identification}\label{sec:anomaly}

It is well known that the static EIT imaging has a fundamental drawback due to the technical difficulties in handling forward modeling errors including the boundary geometry, electrode positions, and other systematic artifacts. Hence, in the anomaly identification problem using EIT, a reference current-voltage data (Neumann-to-Dirichlet data) is required  to  cancel out these common errors by a data subtraction method. Since we can repeat the measurements without and with the mechanical vibration, we can extract the effects of the vibration by taking the difference between two sets of the measured data. In this section, we consider a piecewise constant conductivity  distribution and propose an anomaly location search and parameter estimation algorithm based on the vibration-difference approach.

Let $D=\zt+\epsilon B$ be an anomaly compactly
embedded in $\Omega$, where $\zt$ is a gravitational center of $D$,
$B$ is a $\mathcal{C}^2$-bounded domain containing the origin and
$\epsilon$ is a small positive parameter representing the order of
magnitude of the anomaly size. We suppose that $\sigma$ is locally
homogeneous and $\sigma$ changes abruptly across the boundary of the
anomaly $D$ (experimental results show that the conductivity of the
cancerous tissue is $5$ to $8$ times larger than the one of normal tissue).

We also suppose that the shear modulus ${\mu}$ is piecewise constant
such as
\begin{equation*}
\mu =\left\{ \begin{array}{rl}
\mu_-&\mbox{in}~\Omega\backslash\overline{D},\\
\mu_+&\mbox{in}~D. \end{array}\right.
\end{equation*}
Then the displacement field $\u$ satisfies
\begin{equation}\label{eq:ori_ela}
\left\{\begin{array}{rcll} \,\omega^2 \u+ \mu_-\Delta {\u}
-\nabla p &=&0&\mbox{in}~\Omega\backslash
\overline{D},\\
\,\omega^2\u + \mu_+\Delta \u
-\nabla p &=&0&\mbox{in}~ D,\\
\nabla\cdot {\u}&=&0&\mbox{in}~\Omega, \\
{\u}|_- - {\u}|_+&=&0&\mbox{on}~\partial D,\\
\left(\mu_-\frac{\partial{\u}}{\partial \n}-p\n\right)\Big|_- -
\left(\mu_+\frac{\partial{\u}}{\partial
\n}-p\n\right)\Big|_+&=&0&\mbox{on}~\partial
D,\\
{\u}&=&\g&\mbox{on}~\partial\Omega,
\end{array}\right.
\end{equation}
where
$\pm$ denotes the limit from outside and inside of $D$,
respectively.

Let $\btau_1, \btau_2$ be the tangent vectors at $\p D$ such that
$\{ \btau_1, \btau_2,\n\}$ is an orthonormal basis of
$\mathbb{R}^3$. Our first goal is to provide a representation of
$v_1$ in the case of piecewise constant conductivity distributions.
This can be achieved using layer potential techniques. Similar
arguments to those in \cite{proc_ams, beretta1, beretta2} yield, for
$\x \in \Omega \setminus \overline D$,
\begin{equation} \label{v1poc}
v_1(\x) = - \sigma_+ \int_{\partial D}
\u \cdot \n  \bigg[ \left(1-\frac{\sigma_+}{\sigma_-}\right) \left.\frac{\partial
v}{\partial\n}\right|_+\left.\frac{\partial N}{\partial\n}\right|_+
+ \left(1-\frac{\sigma_-}{\sigma_+} \right)\sum_{j=1}^2\frac{\partial
v}{\partial\btau_j}\frac{\partial N}{\partial\btau_j} \bigg]\, ds,
\end{equation}
where $N$ is the Neumann function given by
$$\left\{ \begin{array}{l}
\na \cdot(\sigma\na N)= \delta_{\bf y} \quad \mbox{in } \Omega,\\
\sigma \frac{\partial N}{\partial \n} = \frac{1}{|\partial \Omega|}   \quad \mbox{on } \partial \Omega,\\
\int_{\partial \Omega} N = 0
\end{array}
\right.
$$
with $|\partial \Omega|$ being the magnitude of surface $\partial
\Omega$.

Now, let $w\in H^1(\Om)$ satisfy $\na\cdot(\sigma\na w)=0$ and let $h$ be defined by
\begin{equation}\label{eq:f2}
h= \sigma_+ \bigg[ \left(1-\frac{\sigma_+}{\sigma_-}\right)
\left.\frac{\partial v}{\partial\n}\right|_+\left.\frac{\partial
w}{\partial\n}\right|_+ + \left(1-\frac{\sigma_-}{\sigma_+}
\right)\sum_{j=1}^2\frac{\partial v}{\partial\btau_j}\frac{\partial
w}{\partial\btau_j} \bigg] \quad \mbox{on } \partial D.
\end{equation}
Note that since the restrictions to $D$ of the solutions to the
conductivity equation $\na\cdot(\sigma\na w)=0$ in $\Omega$ are in
$\mathcal{C}^{1,\alpha}(\overline D)$,  $h \in L^2(\partial D)$. In
order to emphasize the dependence of $v_1$ on $\u$, we denote by
$v_1^\u= v_1$.

The next proposition follows from (\ref{v1poc}) by integration by
parts. It gives the relation between measurable boundary data and
interior information of anomaly $D$.

\begin{proposition}\label{prop:data}
For $w\in H^1(\Om)$ satisfying $\na\cdot(\sigma\na w)=0$, we have
\begin{equation}\label{eq:data_form}
\int_{\partial\Omega} v_1^\u~\sigma\frac{\partial w}{\partial\n} ~ds
=\int_{\partial D}~\u\cdot \n \; h ~ds,
\end{equation}
where $h\in L^2(\partial D)$ is defined by (\ref{eq:f2}).

\end{proposition}

In the sequel, we set
\begin{equation}\label{eq:data_rep}
\eta(\u) := \int_{\partial\Omega}v_1^\u~\sigma\frac{\partial
w}{\partial\n} ~ds.
\end{equation}
The imaging problem is then to locate the anomaly $D$ and to
reconstruct its size, its conductivity, and its shear modulus from
$\eta(\u)$.


\subsection{Location search method and asymptotic
expansion}\label{subsec:position}

In order to have further analysis regarding $\eta(\u)$ in
\eref{eq:data_rep}, we represent the solution $\u$ of
(\ref{eq:ori_ela}) as follows (see \cite{Habib:2011, Habib:2007} for
a detailed derivation and a rigorous statement)
\begin{equation}\label{eq:approx-u}
\u(\x) = \u_0(\x)+\epsilon\v\left(\frac{\x-\zt}{\epsilon}\right)
+ O(\epsilon^2),
\end{equation}
where $\u_0$ is the background displacement field (in the absence of any anomaly)
\begin{equation}\label{eq:hom_ela}
\left\{\begin{array}{rcll}
\omega^2 \u_0+\mu_-\triangle\u_0 -\nabla p_0 &=& 0 & \mbox{in}~\Omega,\\
\nabla\cdot\u_0 &=& 0 &\mbox{in}~\Omega, \\
\u_0 &=& \g &\mbox{on}~\partial\Omega,
\end{array}\right.
\end{equation}
and $\v$ is the solution of
\begin{equation} \left\{\begin{array}{rcl}
\mu_-\triangle\v-\nabla q &=& 0 \quad\mbox{in}~\Bbb R^3\backslash \overline{D},\\
\mu_+\triangle\v-\nabla q &=& 0 \quad\mbox{in}~D, \\
\nabla\cdot\v &=& 0 \quad\mbox{in}~\Bbb R^3 ,\\
\v|_- - \v|_+ &=&0 \quad\mbox{on}~\partial D, \\
\left(\mu_-\frac{\partial\v}{\partial \n}-q\n\right)\Big|_-
- \left(\mu_+\frac{\partial\v}{\partial \n}-q\n\right)\Big|_+
&=& (\mu_--\mu_+)(\nabla\u_0+\nabla\u_0^T)\n \quad\mbox{on}~\partial D,\\
\v(\x) &\rightarrow& 0 \quad\mbox{as}~|\x|\rightarrow+\infty, \\
q(\x) &\rightarrow& 0 \quad\mbox{as}~|\x|\rightarrow+\infty.
\end{array}\right.
\end{equation}
For explicit representations of $\u_0$
and $\v$, let us introduce the fundamental tensor
$\boldsymbol{\Gamma}=(\Gamma_{jk})_{j,k=1}^3$ and $\F=(F_1,F_2,F_3)$
corresponding to the equation
$$
\left(\Delta
+\frac{\omega^2}{\mu_-}\right)\Gamma_{jk}(\x)-\partial_k
F_j(\x) = \delta_{jk}\delta(\x)\quad\mbox{in}~ \R^3
$$
and $\nabla\cdot\bGamma=\mathbf{0}$, where
\begin{eqnarray}\label{eq:fund_Stokes}
\Gamma_{jk}(\x)&=&-\frac{\delta_{jk}}{4\pi}\frac{e^{i\frac{\omega}{\sqrt{\mu_-}}|\x|}}{|\x|}
-\frac{\mu_-}{4\pi \omega^2}\partial_j\partial_k\frac{e^{i\frac{\omega}{\sqrt{\mu_-}}|\x|}-1}{|\x|}\\
F_j(\x)&=&-\frac{1}{4\pi}\frac{x_j}{|\x|^3}.\nonumber
\end{eqnarray}
Define $\rm O_R=\{\y:|\y|\leq R, \,\, R~\mbox{sufficiently large}\}$
such that $\overline{\Omega}\subset \rm{O}_R$. If
$\g(\x)=\bGamma(\x-\bar\y)\bq$ with direction of the wave $\bq$ for
a point source $\bar\y\in\partial \rm{O}_R$, then we have
$p_0(\x)=\F(\x-\bar\y)\cdot\bq$ and
\begin{eqnarray*}
&\u_0(\x) = \frac{1}{\mu_-}\bGamma(\x-\bar\y)\,\bq,\\
&\v\left(\frac{\x-\zt}{\epsilon}\right)   =
\mathcal{S}_{B}^0\left(-\frac{\mu_- +\mu_+}{2(\mu_- -\mu_+)} {\bf I}
+\left(\mathcal{K}_{B}^0\right)^\ast\right)^{-1}
\left[\frac{\partial \u_0}{\partial
\n}(\z^*)\right]\left(\frac{\x-\zt}{\epsilon}\right),
\end{eqnarray*}
where $\mathcal{S}_{B}^0$ is a single layer potential for the Stokes
system, $\mathcal K_B^0$ is the boundary integral operator and
$(\mathcal{K}_{B}^0)^*$ is the $L^2$-adjoint operator of $\mathcal
K_B^0$ with superscript $0$ standing for the static case $\omega=0$.

Noting that $\u$ is depending on $\bq$ and the point source
$\bar\y$, we can denote $\u$ by $\u_{\bq,\bar\y}$.
%
Define $J:\Omega\rightarrow\Bbb R$ by
\begin{equation}\label{eq:J}
J(\z^s) := \sum_{j=1}^3\int_{\partial \rm O_R}
\k^T\,\bGamma(\z^s-\bar\y)\,\bq_j~\overline{\eta(\u_{\bq_j,\bar\y})}~ds_{\bar\y}
\end{equation}
for three orthonormal vectors $\bq_j~(j=1,2,3)$, $\z^s\in \Omega$
and a constant unit vector $\k$. Here, $\z^s$ is considered as a
searching point in $\Omega$.

The following lemma follows from the Helmholtz-Kirchhoff identity
for $\bGamma$.
\begin{lemma}
The functional $J(\z^s)$ can be estimated by
\begin{equation}\label{eq:411}
J(\z^s) = \k^T\,{{\rm Im}}(\bGamma(\z^s-\zt))\,\h
+\mathcal{O}(\epsilon^3(\omega^2+1+R^{-1})),
\end{equation}
where $\h=\frac{1}{\mu_-}\int_{\partial D} h ~\n  \, ~ds$.
\end{lemma}
\begin{proof}
Using \eref{eq:approx-u} and the above representations, the
functional $\eta(\u)$ can be written as
\begin{eqnarray*}
\eta(\u) &=& \int_{\partial D}~\u(\x)\cdot\n \;  h(\x) ~ds_\x\\
&=&\int_{\partial D}\u_0(\x)\cdot\n \;  h(\x)
~ds_\x +\mathcal{O}\left(\frac{\epsilon^3(\omega^2+1)}{R^2}\right).
\end{eqnarray*}
Then we can write $J(\z^s)$ as
\begin{eqnarray*}
J(\z^s)&=& \sum_{j=1}^3\int_{\partial \rm O_R}~\k^T\bGamma(\z^s-\y)
\bq_j~\overline{\eta(\u_{\bq_j,\y})}~ds_\y \\
&=& \sum_{j=1}^3\int_{\partial \rm O_R}~\k^T\bGamma(\z^s-\y)
\bq_j~\overline{\int_{\partial D}~\u_{\bq_j,\y}\cdot\n\; h
~ds}~ds_\y \\
&=&\frac{1}{\mu_-}\int_{\partial D}\k^T \left[\int_{\partial \rm
O_R} \bGamma(\z^s-\y)\overline{\bGamma(\zt-\y)} \,ds_\y\right] \n \;
h(\x)\,ds_\x +\mathcal{O}(\epsilon^3(\omega^2+1)).
\end{eqnarray*}
From the Sommerfeld radiation condition \cite{Habib:2013:Book}, we
have
\begin{eqnarray*}
&&\hspace{-1cm}\int_{\partial \rm O_R}
\bGamma(\z^s-\y)\overline{\bGamma(\zt-\y)}
\,ds_\y \\
&&\hspace{-1cm}= \int_{\partial \rm O_R}
\left(\frac{\partial}{\partial\n}\bGamma(\z^s-\y)
\overline{\bGamma(\zt-\y)}-\bGamma(\z^s-\y)\frac{\partial}{\partial
\n}\overline{\bGamma(\zt-\y)}\right)
ds_\y+\mathcal{O}\left(\frac{1}{R}\right).
\end{eqnarray*}
Since $|\z^s|,~|\zt|<R$, the property of fundamental solution and
Green's identity imply
\begin{equation*}
\int_{\partial \rm O_R}
\left(\frac{\partial}{\partial\n}\bGamma(\z^s-\y)
\overline{\bGamma(\zt-\y)}-\bGamma(\z^s-\y) \frac{\partial}{\partial
\n}\overline{\bGamma(\zt-\y)}\right) ds_\y = 2i\,{\rm Im}
(\bGamma(\z^s-\zt)).
\end{equation*}
Hence we have
$$
J(\z^s) =\frac{1}{\mu_-} \k^T{{\rm Im}}(\bGamma(\z^s-\zt))
\int_{\partial D}\, h \; \n \,ds
+\mathcal{O}\left({\epsilon^3}{R^{-1}}\right)+\mathcal{O}(\epsilon^3(\omega^2+1)).
$$
\end{proof}

Since ${\rm Im}(\bGamma(\z^s-\z))$ has a sinc function as a
component, $J$ has its maximum at $\z^s=\zt$. In $J$, the
fundamental tensor of the Stokes problem, $\bGamma$, can be replaced
by a simple exponential function. Using that the following
proposition proposes an approximation of $J(\z^s)$ which is more
practical
for finding the maximum and hence locating the anomaly. 

\begin{proposition}\label{prop:position}
Define $\tilde J$ by
\begin{equation}\label{eq:tJ}
\tilde J (\z) :=\int_{\partial \rm O_R}
e^{i\omega\sqrt{\frac{1}{\mu_-}}|\z-\y|}\overline{\eta(\u_{\bq,\y})}
\,ds_\y.
\end{equation}
Then the point $\z^s\in\Omega$ satisfying $ \z^s=\arg\max_{\z\in
\Omega}\tilde J(\z) $ is the center position of $D$.
\end{proposition}
\begin{proof}
From the definition of $\bGamma$ in (\ref{eq:fund_Stokes}), we have
\begin{eqnarray*}
-4\pi
\Gamma_{jk}(\x)&=&\delta_{jk}\frac{e^{i\omega\sqrt{\frac{1}{\mu_-}}|\x|}}{|\x|}
+\frac{\mu_-}{\omega^2}\partial_j\partial_k\frac{e^{i\omega\sqrt{\frac{1}{\mu_-}}|\x|}-1}{|\x|}\\
&=&\delta_{jk}\frac{e^{i\omega\sqrt{\frac{1}{\mu_-}}|\x|}}{|\x|}+O(|\x|^{-3}).
\end{eqnarray*}
If $|\z^s-\y|\rightarrow\infty$, then the following approximation holds
\[
\bGamma(\z^s-\y)\approx
-\frac{1}{4\pi}\frac{e^{i\omega\sqrt{\frac{1}{\mu_-}}|\z^s-\y|}}{|\z^s-\y|} {\bf I}
\]
with ${\bf I}$ the identity matrix. Here, the phase terms of $J$ and
$\tilde J$ are the same for the identity matrix $I$. In fact, the
phase term of $\tilde J$ is
\begin{eqnarray*}
&&\int_{\partial D}\int_{\partial \rm O_R}
e^{i\omega\sqrt{\frac{1}{\mu_-}}|\z-\y|}\overline{e^{i\omega\sqrt{\frac{1}{\mu_-}}|\zt-\y|}}ds_\y
ds_\x \\
&&\quad\approx  \int_{\partial
D}\int_{\partial\Omega}e^{i\omega\sqrt{\frac{1}{\mu_-}}
\left(|\y|-\frac{\z\cdot\y}{|\y|}\right)}{e^{-i\omega\sqrt{\frac{1}{\mu_-}}
\left(|\y|-\frac{\zt\cdot\y}{|\y|}\right)}}ds_\y
ds_\x\\
&&\quad= \int_{\partial D}\int_{\partial
\rm O_R}e^{i\omega\sqrt{\frac{1}{\mu_-}}\left(\frac{(\zt-\z)\cdot\y}{|\y|}\right)}ds_\y
ds_\x= \int_{\partial D}\int_{\partial
\rm O_1}e^{i\omega\sqrt{\frac{1}{\mu_-}}(\zt-\z)\cdot\hat{\y}}ds_{\hat{\y}}
ds_\x\\
&&\quad\approx\int_{\partial
D}\frac{\sqrt{\mu_-}}{\omega}\frac{\sin\left(\frac{\omega|\zt-\z|}{\sqrt{\mu_-}}\right)}{
|\zt-\z|}\,ds_\x
\end{eqnarray*}
in which the first approximation holds because
$|\y|,|\y-\z|,|\y-\zt|\gg 1$ imply that the angles between them are
close to 0. Therefore, it has its maximum at $|\zt-\z|=0$.
\end{proof}

Proposition \ref{prop:position} shows that the conductivity anomaly
can be detected with a resolution of the order of half the elastic
wavelength.

\subsection{Size estimation and  reconstruction of the material parameters}\label{subsec:mu}

In the previous subsection, a formula to find the center position
$\zt$ of $D$ has been proposed. Here, we propose a method to estimate the size $\delta$,
the conductivity $\sigma_+$, and the shear modulus $\mu_+$ of the anomaly $D$.
For computational simplicity, we assume that $D$ is a sphere, the
background conductivity, $\sigma_-$, and shear modulus, $\mu_-$, are
known.

Using a broadband frequency range for elastic vibrations, we can
acquire time-domain data corresponding to $\g(\x,t)=
\check{\bGamma}(\x-\y,t)\bq$ for $\y\in\partial\Omega$. Here,
$\check{\bGamma}$ is the inverse Fourier transform taken in $\omega$
variable of the fundamental solution $\bGamma$ to the Stokes system.
Take $w= v$ in (\ref{eq:data_rep}) and rewrite $\eta$ as a function
of time $t$. It follows that
\begin{equation}\label{eq:data_t}
\eta(t) =\int_{\partial D}\u(\x,t)\cdot\n(\x)\,h(\x) ~ds_\x,
\end{equation}
where
$$
h = \sigma_+ \bigg[
\left(1-\frac{\sigma_+}{\sigma_-}\right) \left(\left.\frac{\partial
v}{\partial\n}\right|_+\right)^2  +
(1-\frac{\sigma_-}{\sigma_+})\sum_{j=1}^2\left(\frac{\partial
v}{\partial\btau_j}\right)^2 \bigg].
$$
Define $t_\y^a$ and $t_\y^b$ by
\begin{equation*}
\begin{array}{lll}
t_\y^a&:=&\mbox{the first $t$ such that }\eta(t)\neq 0\\
&=& \mbox{the first $t$ such that a sphere of center $\y$ and growing radius hits}~\partial D,~\mbox{say }\z_a\\
&=& {|\z_a-\y|}/{\sqrt{\mu_-}}\\
t_\y^b&:=&\mbox{the last $t$ such that }\eta(t)\neq 0\\
&=& \mbox{the last $t$ such that a sphere of center $\y$ and growing radius hits }\partial D,~\mbox{say }\z_b\\
&=&  {|\z_b-\y|}/{\sqrt{\mu_-}}.
\end{array}
\end{equation*}
Then the radius, $\delta$, of $D$ can be estimated by
\begin{equation}\label{eq:rel_r_mu}
2 \delta \approx \sqrt{\mu_+}\left(t_\y^b-t_\y^a\right).
\end{equation}
If we know the size of $D$ then we can extract $\mu_+$ information.
If not, we need further investigation as follows. We propose to minimize over
$\mu_+$ and $\sigma^+$ the following discrepancy functional
$$
\int_{t_{\y^a}}^{t_{\y^b}} \bigg| \eta(t) - \check{\bGamma}(\x-\z^*,t) \bq \cdot \int_{\partial D}
\n \, h \, ds \bigg|^2.
$$
To compute $v$ we use relation (\ref{eq:rel_r_mu}) and the fact that
$D= \z^* + \delta B$ with $B$ being the unit sphere centered at the
origin.

\section{Numerical Simulations}

First, we will present numerical simulation results showing voltage
differences when the mechanical vibration is applied. The second numerical simulation will show reconstructed images using the algorithm proposed in section \ref{sec:PIR}. Then, we will show a
numerical evidence of the position finding formula proposed in subsection \ref{subsec:position}.

\subsection{Simulations of the voltage difference
map}\label{sec:NA_vol}

We present two results of numerical simulations to show the voltage
difference map of $v_1$ before and after the applied mechanical
vibration. We consider a cubic container as shown in Table
\ref{fig:model}. The sensing (measuring) electrodes are placed at
the bottom of the container and the sinusoidal mechanical vibration
is applied through the top surface, which is also the current
driving electrode. In the second numerical test, the mechanical
vibration is applied through the lateral surface. Two anomalies, a
small spherical anomaly and a large cylindrical anomaly, are placed
in the container with different material properties shown in Table
\ref{table:sim_1}.
\begin{figure}[ht!]
\begin{center}
\includegraphics[scale=0.4]{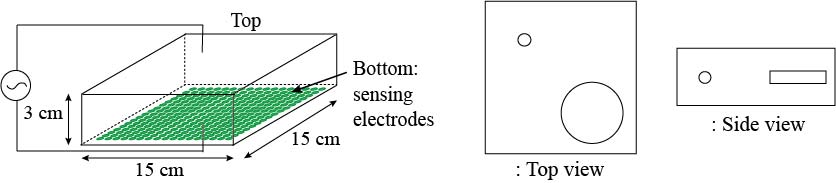}
\caption{Model for numerical simulations.}\label{fig:model}
\end{center}
\end{figure}

\begin{table}[ht!]
\centering
\begin{tabular}{|c|c|c|}
\hline
{} & Background &Anomalies  \\
\hline
 Shear modulus  & 0.266 & 2.99   \\
\hline
\end{tabular}\caption{Shear modulus values used in numerical simulations.}\label{table:sim_1}
\end{table}

Figures \ref{fig:diff_v} presents the measured voltage difference
$v_1$ at the bottom surface. It clearly shows the perturbation of the
conductivity distribution inside $\Omega$ caused by the mechanical vibration.
\begin{figure}[ht!]
\begin{center}
\includegraphics[scale=0.4]{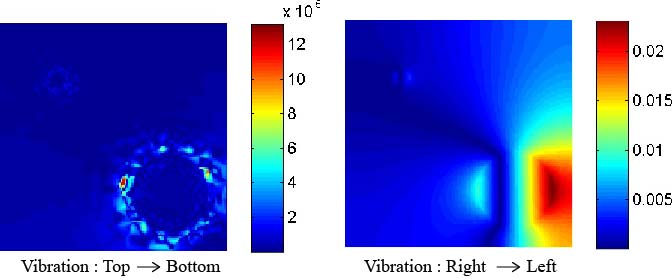}
\caption{Maps of the voltage difference $v_1$ on the bottom surface subject to two different mechanical vibrations.}\label{fig:diff_v}
\end{center}
\end{figure}

\subsection{Image reconstruction} \label{sec:IR}

Based on the analysis in section \ref{sec:model} and the numerical
evidence shown in subsection \ref{sec:NA_vol}, we reconstruct the
perturbed conductivity distribution caused by the mechanical
vibration using the algorithm introduced in section \ref{sec:PIR}.
Let $\Omega$ be a cubical domain such as in Figure \ref{fig:Logan}
with size $15\times 15 \times 3~\mbox{cm}^3$. Since we need two sets
of measured voltage data when the current is injected through two
different directions (the directions of current flow should be
linearly independent each other), respectively, we employ the
positions of driving electrodes as in Figure \ref{fig:new_cont} that
use the whole lateral side as current injecting surface. Sensing
electrodes in $256\times 256$ -array are placed on the bottom
surface. Inside of the container, we placed eight objects as in
Figure \ref{fig:Logan}. The admittivity of the background was
$0.01+i0.01$ and the anomalies had $0.03+i0.03$ with $f=10$ kHz. The
material coefficients for the elasticity equation were set to be the
same as in Table \ref{table:sim_1}. We applied the mechanical
vibration with the frequency of 100 Hz. Figure \ref{fig:PIR1} shows
the reconstructed images of the projected distribution of
$\u\cdot\nabla \ln \sigma$ on the the sensing surface $\Upsilon$ of
the model in Figure \ref{fig:Logan} when the mechanical vibration is
applied from the (a) top, (b) left, and (c) front sides,
respectively.

\begin{figure}[ht!]
\begin{center}
\includegraphics[height=3.5cm,width=13cm]{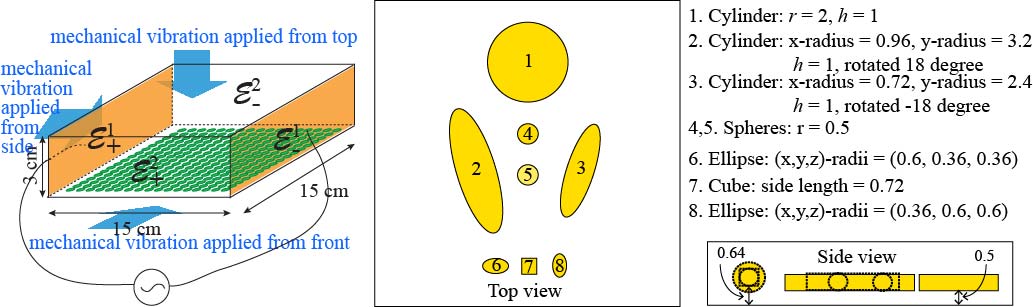}
\caption{Model for image reconstructions: `$r$' is the radius and
`$h$' is the thickness of the object in cm.}\label{fig:Logan}
\end{center}
\end{figure}

\begin{figure}[ht!]
\begin{center}
\includegraphics[scale=0.32]{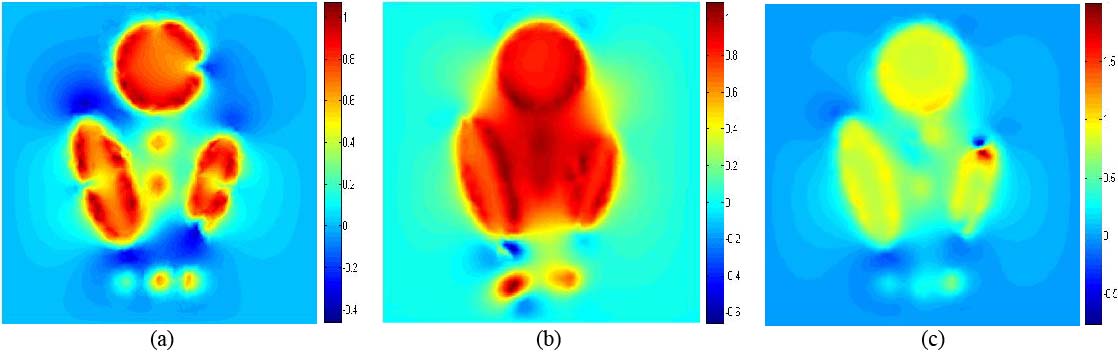}
\caption{Reconstructed images of the model in Figure \ref{fig:Logan}: mechanical vibration
is applied from the (a) top surface, (b) side wall and (c) front wall.}\label{fig:PIR1}
\end{center}
\end{figure}

\subsection{Application to TAM} \label{sec:TAM_app}

In this numerical test, we use the same configuration used in
\cite{PEIT:2013} as shown in Figure \ref{fig:model_TAM}. The
computational domain is enlarged to the size of $19\times 19 \times
3~\mbox{cm}^3$ in order to attach current injecting electrodes while
all the objects are at the same position as in Figure
\ref{fig:Logan} with same material properties. Now current injecting
electrodes are placed on the bottom surface which are still
separated from voltage measuring sensing electrodes with size of
$2\times 2~\mbox{cm}^2$.

\begin{figure}[ht!]
\begin{center}
\includegraphics[scale=0.4]{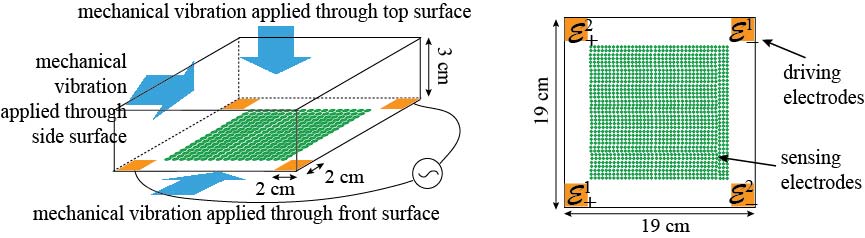}
\caption{Model for numerical simulations based on TAM
model.}\label{fig:model_TAM}
\end{center}
\end{figure}

Also the mechanical vibration with the frequency of 100 Hz is
applied through the top, side, and front surfaces. The reconstructed
images are presented in Figure \ref{fig:recon_TAM}.

\begin{figure}[ht!]
\begin{center}
\includegraphics[scale=0.32]{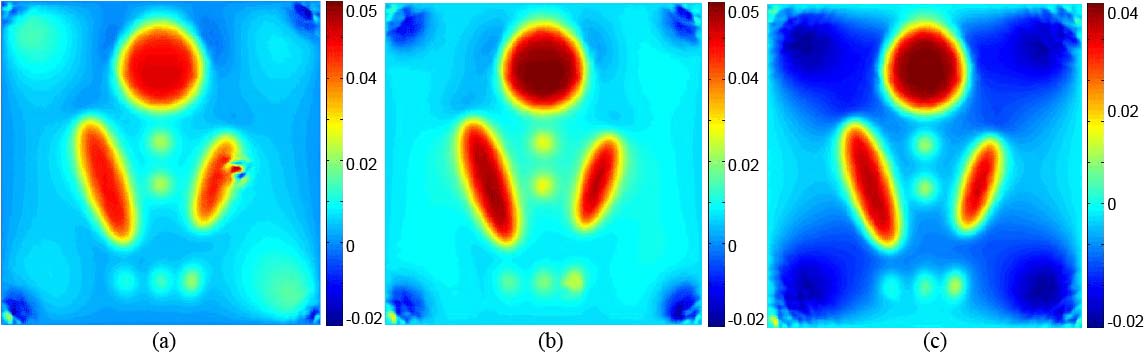}
\caption{Reconstructed images using data collected from the model in
Figure \ref{fig:model_TAM}: mechanical vibration is applied from the
(a) top surface, (b) side wall and (c) front
wall.}\label{fig:recon_TAM}
\end{center}
\end{figure}

\subsection{Anomaly location}\label{sec:position}

In subsection \ref{subsec:position}, the formula (\ref{eq:411}) and
Proposition \ref{prop:position} suggest to seek the maximizer of $J$
or $\tilde J$ to locate the center position of anomaly. To verify
Proposition \ref{prop:position}, we consider a cylindrical domain
$\Omega$ centered at (0,0,1.5) with radius 7.5 cm and height 3 cm.
Let the anomaly $D$ be a sphere with radius 0.25 cm, centered at
$\z^*=(3.75,0,1.5)$. As shown in subsection \ref{subsec:position},
the displacement $\u$ depends on $\bq$ and the point source $\y$.
Here, $\bq$ is set to $(1,0,0),\,(0,1,0)$ and $(0,0,1)$ and the
point source $\y_k$ is chosen for $k=1,\cdots,1940$ which are
uniformly distributed on $\rm O_R$, a sphere centered at (0,0,1.5)
with radius 37.5 cm so that $\overline{\Omega}\subset \rm{O}_R$. The
Figure \ref{fig:point_s} shows the computed discrete version of
$\tilde J (\z)$ for each $\z=(x,y,1.5)$ as follows
\begin{equation*}
\tilde J (\z) :=\int_{\partial \rm O_R}
e^{i\omega\sqrt{\frac{1}{\mu_-}}|\z-\y|}\overline{\eta(\u_{\bq,\y})}
\,ds_\y\approx\sum_{k=1}^{1940}
e^{i\omega\sqrt{\frac{1}{\mu_-}}|\z-\y_k|}\overline{\eta(\u_{\bq,\y_k})},
\end{equation*}
where $\mu_-=1$ and $\omega =200\times \pi$.

\begin{figure}[ht!]
\centering
\includegraphics[scale=0.20]{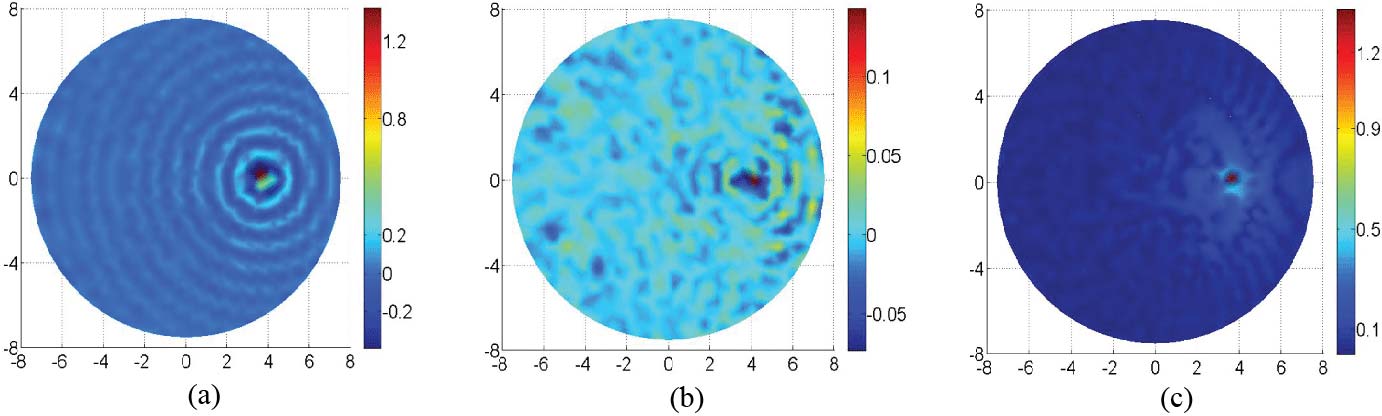}
\caption{$\tilde J(\z)$ when $\z=(x,y,1.5)$: The maximum point is
the center position of $D$ which is (3.75,0,1.5): (a) {\rm
Re($J(\z)$)} (b) {\rm Im($J(\z)$)} (c) $|J(\z)|$}\label{fig:point_s}
\end{figure}

Figure \ref{fig:point_s} shows that the formula proposed in
Proposition \ref{prop:position} finds the center position of anomaly
$D$ under an ideal circumstance such that no noise is added and all
mathematical assumptions are satisfied. An analysis, in the same
spirit of \cite{jugnon}, of the statistical
 stability with respect to medium and measurement noises of the
 localization algorithm will be the subject of a forthcoming work.

\section{Concluding remarks}

In this paper, we proposed a new multi-physics electrical impedance imaging
approach using mechanical vibrations simultaneously applied to an imaging object together with current injections. We provided the mathematical framework for the proposed approach and
presented a few numerical simulation results to illustrate its resolution and stability.

It is worth mentioning that the proposed approach can also be used to measure the elasticity of an internal object with known electrical conductivity values. Using the electrical conductivity image, one can reconstruct the displacement field at the scale of the changes of the conductivity and then, recover the shear modulus. This will be the subject of a forthcoming
publication.

\end{document}